\documentclass[final,copyright]{eptcs}

\providecommand{\event}{QAPL 2012}

\usepackage[all]{xy}

\usepackage{amssymb}

\newtheorem{theorem}{Theorem}
\newtheorem{proposition}[theorem]{Proposition}
\newtheorem{definition}[theorem]{Definition}
\newtheorem{example}[theorem]{Example}
\newtheorem{corollary}[theorem]{Corollary}

\newenvironment{proof}{\begin{rm}
                       \begin{list}{}{\setlength{\leftmargin}{0pt}
                                      \setlength{\rightmargin}{0pt}
                                      \setlength{\labelwidth}{0pt}
                                      \setlength{\itemindent}{0pt}
                                      \setlength{\topsep}{8pt}
                                      \setlength{\parskip}{0pt}
                                      \setlength{\partopsep}{0pt}}
                       \item[]{\bf Proof\ }}%
                      {\end{list}
                       \end{rm}}

\def\qed{\ifmmode
         $\square$
         \else
         {\unskip
          \nobreak
          \hfil
          \penalty50
          \hskip1em
          \null
          \nobreak
          \hfil
          $\square$
          \parfillskip=0pt
          \finalhyphendemerits=0
          \endgraf}
         \fi}

\newcommand{\comment}[1]{\hspace{1em}\mbox{[{\small #1}]}}

\mathcode`<="4268
\mathcode`>="5269
\mathchardef\ls="213C
\mathchardef\gr="213E

\newcommand{\always}{\square}
\newcommand{\eventually}{\lozenge}
\newcommand{\until}{\mathbin{\mathcal{U}}}
\newcommand{\wuntil}{\mathbin{\mathcal{W}}}
\newcommand{\release}{\mathbin{\mathcal{R}}}
\newcommand{\nxt}{\bigcirc}

\newcommand{\AP}{\mathchoice{\mbox{\it AP}}{\mbox{\it AP}}{\mbox{\it \scriptsize AP}}{\mbox{\it \tiny AP}}}
\newcommand{\source}{\mathchoice{\mbox{\rm source}}{\mbox{\rm source}}{\mbox{\rm \scriptsize source}}{\mbox{\rm \tiny source}}}
\newcommand{\target}{\mathchoice{\mbox{\rm target}}{\mbox{\rm target}}{\mbox{\rm \scriptsize target}}{\mbox{\rm \tiny target}}}
\newcommand{\pro}{\mathchoice{\mbox{\rm prob}}{\mbox{\rm prob}}{\mbox{\rm \scriptsize prob}}{\mbox{\rm \tiny prob}}}
\newcommand{\lab}{\mathchoice{\mbox{\rm label}}{\mbox{\rm label}}{\mbox{\rm \scriptsize label}}{\mbox{\rm \tiny label}}}
\newcommand{\trace}{\mathchoice{\mbox{\rm trace}}{\mbox{\rm trace}}{\mbox{\rm \scriptsize trace}}{\mbox{\rm \tiny trace}}}

\newcommand{\Exec}{\mathchoice{\mbox{\rm Exec}}{\mbox{\rm Exec}}{\mbox{\rm \scriptsize Exec}}{\mbox{\rm \tiny Exec}}}
\newcommand{\prog}{\mathchoice{\mbox{\rm prog}}{\mbox{\rm prog}}{\mbox{\rm \scriptsize prog}}{\mbox{\rm \tiny prog}}}
\newcommand{\pref}{\mathchoice{\mbox{\rm pref}}{\mbox{\rm pref}}{\mbox{\rm \scriptsize pref}}{\mbox{\rm \tiny pref}}}
\newcommand{\out}{\mathchoice{\mbox{\rm out}}{\mbox{\rm out}}{\mbox{\rm \scriptsize out}}{\mbox{\rm \tiny out}}}

\newcommand{\Nset}{\mathbb{N}}

\newcommand{\E}{\mathit{E}}
\newcommand{\Emin}{{\mathit{ME}}}

\title{Measuring Progress of Probabilistic LTL Model Checking}
\author{Elise Cormie-Bowins\thanks{Supported by an Ontario Graduate
    Scholarship.}
\hspace{0.2em} and Franck van Breugel\thanks{Supported by the Natural Sciences and Engineering Research Council of Canada and the Leverhulme Trust.}
\institute{DisCoVeri Group, Department of Computer Science, York University,\\
4700 Keele Street, Toronto, ON, M3J 1P3, Canada}}
\def\titlerunning{Measuring Progress of Probabilistic LTL Model Checking}
\def\authorrunning{E. Cormie-Bowins and F. van Breugel}

\begin{document}

\maketitle

\begin{abstract}
Recently, Zhang and Van Breugel introduced the notion of a progress
measure for a probabilistic model checker.  Given a linear-time 
property $\phi$ and a description of the part of the system that has 
already been checked, the progress measure returns
a real number in the unit interval.  The real number captures how much
progress the model checker has made towards verifying $\phi$.  If
the progress is zero, no progress has been made.  If it is one, the
model checker is done.  They showed that the progress measure provides
a lower bound for the measure of the set of execution paths that
satisfy $\phi$.  They also presented an algorithm to compute the 
progress measure when $\phi$ is an invariant.

In this paper, we present an algorithm to compute the progress
measure when $\phi$ is a formula of a positive fragment of linear
temporal logic.  In this fragment, we can express invariants
but also many other interesting properties.  The algorithm is
exponential in the size of $\phi$ and polynomial in the size of 
that part of the system that has already been checked.  We
also present an algorithm to compute a lower bound for the 
progress measure in polynomial time.
\end{abstract}

\section{Introduction}

Due to the infamous state space explosion problem, model
checking a property of source code that contains randomization 
often fails.  In many cases, the probabilistic model checker
simply runs out of memory without reporting any useful
information.  In \cite{ZB11:icalp}, Zhang and Van Breugel
proposed a progress measure for probabilistic model checkers.
This measure captures the amount of progress the model
checker has made with its verification effort.  Even if
the model checker runs out of memory, the amount of progress
may provide useful information.

Our aim is to develop a theory that is applicable to
probabilistic model checkers in general.  Our initial
development has been guided by a probabilistic extension
of the model checker Java PathFinder (JPF) \cite{VHBPL03:ase}.  
This model checker can check properties, expressed in 
linear temporal logic (LTL), of Java code containing 
probabilistic choices.

We model the code under verification as a probabilistic
transition system (PTS), and the systematic search of the
system by the model checker as the set of explored
transitions of the PTS.  We focus on linear-time
properties, in particular those expressed in LTL.  The
progress measure is defined in terms of the set of
explored transitions and the linear-time property
under verification.  The progress measure returns a
real number in the interval $[0, 1]$.  The larger this
number, the more progress the model checker has made with
its verification effort.

Zhang and Van Breugel showed that their progress measure
provides a lower bound for the measure of the set of
execution paths that satisfy the linear-time property under 
verification.  If, for example, the progress is 0.9999,
then the probability that we encounter a violation of
the linear-time property when we run the code is at
most 0.0001.  Hence, despite the fact the model checker
may fail by running out of memory, the verification effort 
may still be a success by providing an acceptable upper bound
on the probability of a violation of the property.

The two main contributions of this paper are
\begin{enumerate}
\item
a characterization of the progress measure for a positive
fragment of LTL.  This fragment
includes invariants, and most examples found in, for example,
\cite[Section~5.1]{BK08} can be expressed in this fragment.
This characterization forms the basis for an algorithm
to compute the progress measure.
\item
a polynomial time algorithm to compute a lower bound for
the progress measure for the positive fragment of LTL.
The lower bound is tight for invariants, that is, this
algorithm computes the progress for invariants.
\end{enumerate}

\section{A Progress Measure}

In this section, we review some of the key notions and results of 
\cite{ZB11:icalp}.  We represent the system to be verified by the
probabilistic model checker as a probabilistic transition system.

\begin{definition}
A {\sl probabilistic transition system} is a tuple
$<S, T, \AP, s_0, \source, \target, \pro, \lab>$ 
consisting of
\begin{itemize}
\item
a countable set $S$ of states,
\item
a countable set $T$ of transitions,
\item
a set $\AP$ of atomic propositions,
\item
an initial state $s_0$,
\item
a function $\source : T \to S$,
\item
a function $\target : T \to S$,
\item
a function $\pro : T \to (0, 1]$, and
\item
a function $\lab : S \to 2^{\AP}$
\end{itemize}
such that
\begin{itemize}
\item
$s_0 \in S$ and
\item
for all $s \in S$, $\sum \{\, \pro(t) \mid \source(t) = s \,\} = 1$.
\end{itemize}
\end{definition}

\begin{example}
\label{example:first}
The probabilistic transition system $\mathcal{S}$ depicted by
$$
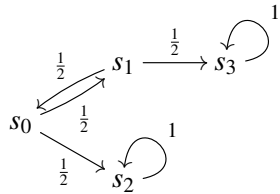
\UseComputerModernTips
\xymatrix@R=2ex{
& s_1 \ar[r]^{\frac{1}{2}} \ar@/_0.2pc/[dl]_{\frac{1}{2}} & s_3 \ar@(r,u)[]_{1}\\
s_0 \ar@/_0.2pc/[ur]_{\frac{1}{2}} \ar[dr]_{\frac{1}{2}}\\
& s_2 \ar@(r,u)[]_{1}
}
$$
has three states and six transitions.  In this example, we use the
indices of the source and target to name the transitions.  For example, the
transition from $s_0$ to $s_2$ is named $t_{02}$.  Given this naming
convention, the functions $\source_{\mathcal{S}}$ and $\target_{\mathcal{S}}$
are defined in the obvious way.  For example,
$\source_{\mathcal{S}}(t_{02}) = s_0$ and
$\target_{\mathcal{S}}(t_{02}) = s_2$.
The function $\pro_{\mathcal{S}}$ can be easily extracted from the above
diagram.  For example, $\pro_{\mathcal{S}}(t_{02}) = \frac{1}{2}$.
All states are labelled with the atomic proposition $a$ and the states
$s_1$ and $s_2$ are also labelled with the atomic proposition $b$.
Hence, for example, $\lab_{\mathcal{S}}(s_2) = \{ a, b \}$.
\end{example}

Instead of $<S, T, \AP, s_0, \source, \target, \pro, \lab>$ we usually write
$\mathcal{S}$ and we denote, for example, its set of states by 
$S_{\mathcal{S}}$.  We model the potential executions of the system
under verification as execution paths of the PTS.

\begin{definition}
An {\sl execution path} of a PTS $\mathcal{S}$ is an infinite sequence 
of transitions $t_1 t_2 \ldots$ such that
\begin{itemize}
\item
for all $i \geq 1$, $t_i \in T_{\mathcal{S}}$,
\item
$\source_{\mathcal{S}}(t_1) = {s_0}_{\mathcal{S}}$, and
\item
for all $i \geq 1$, $\target_{\mathcal{S}}(t_i) = \source_{\mathcal{S}}(t_{i+1})$.
\end{itemize}
The set of all execution paths is denoted by $\Exec_{\mathcal{S}}$.
\end{definition}

\begin{example}
Consider the PTS of Example~\ref{example:first}.  For this system,
$t_{02} {t_{22}}^{\omega}$, $t_{01} t_{13} {t_{33}}^{\omega}$, and
$t_{01} t_{10} t_{02} {t_{22}}^{\omega}$ are examples of execution
paths.
\end{example}

To define the progress measure, we use a measurable space of
execution paths.  We assume that the reader is familiar with
the basics of measure theory as can be found in, for example,
\cite{B95}.  Recall that a measurable space consists of a set, a
$\sigma$-algebra and a measure.  In our case, the set is 
$\Exec_{\mathcal{S}}$.  The $\sigma$-algebra $\Sigma_{\mathcal{S}}$
is generated from the basic cylinder sets defined below.  We 
denote the set of finite prefixes of execution paths in 
$\Exec_{\mathcal{S}}$ by $\pref(\Exec_{\mathcal{S}})$.

\begin{definition}
Let $e \in \pref(\Exec_{\mathcal{S}})$.  Its {\sl basic cylinder set} 
$B_{\mathcal{S}}^e$ is defined by
\begin{displaymath}
B_{\mathcal{S}}^e 
= 
\{\, e' \in \Exec_{\mathcal{S}} \mid  e \mbox{ is a prefix of } e' \,\}.
\end{displaymath}
\end{definition}

The measure $\mu_{\mathcal{S}}$ is defined on a basic cylinder set 
$B_{\mathcal{S}}^{t_1 \ldots t_n}$ by
\begin{displaymath}
\mu_{\mathcal{S}}(B_{\mathcal{S}}^{t_1 \ldots t_n}) 
=
\prod_{1 \leq i \leq n} \pro_{\mathcal{S}}(t_i). 
\end{displaymath}
The measurable space 
$<\Exec_{\mathcal{S}}, \Sigma_{\mathcal{S}}, \mu_{\mathcal{S}}>$
is a sequence space as defined, for example, in \cite[Chapter~2]{KLK66}.  

The verification effort of the probabilistic model checker is
represented by its search of the PTS.  The search is captured
by the set of transitions that have been explored during the
search.

\begin{definition}
A {\sl search} of a PTS $\mathcal{S}$ is a finite subset of $T_{\mathcal{S}}$.
\end{definition}

\begin{example}
Consider the PTS of Example~\ref{example:first}.  The sets
$\emptyset$, $\{ t_{01} \}$, $\{ t_{02} \}$, $\{ t_{01}, t_{02} \}$
and $\{ t_{01}, t_{02}, t_{10}, t_{13}, t_{22}, t_{33} \}$ are
examples of searches.
\end{example}

A PTS is said to extend a search if the transitions of the search 
are part of the PTS.  We will
use this notion in the definition of the progress measure.

\begin{definition}
The PTS $\mathcal{S}'$ {\sl extends} the search $T$ of the 
PTS $\mathcal{S}$ if for all $t \in T$,
\begin{itemize}
\item
$t \in T_{\mathcal{S}'}$,
\item
${s_{0}}_{\mathcal{S}}  = {s_{0}}_{\mathcal{S'}}$,
\item
$\source_{\mathcal{S}'}(t) = \source_{\mathcal{S}}(t)$,
\item
$\target_{\mathcal{S}'}(t) = \target_{\mathcal{S}}(t)$,
\item
$\pro_{\mathcal{S}'}(t) = \pro_{\mathcal{S}}(t)$,
\item
$\lab_{\mathcal{S}'}(\source_{\mathcal{S}'}(t)) =   \lab_{\mathcal{S}}(\source_{\mathcal{S}}(t))$, and
\item
$\lab_{\mathcal{S}'}(\target_{\mathcal{S}'}(t)) = \lab_{\mathcal{S}}(\target_{\mathcal{S}}(t))$.
\end{itemize}
\end{definition}

\begin{example}
Consider the PTS of Example~\ref{example:first} and the search
$\{ t_{01}, t_{02} \}$.  The PTS
$$
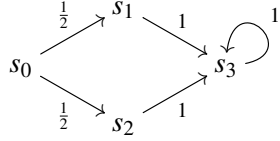
\UseComputerModernTips
\xymatrix@R=2ex{
& s_1 \ar[dr]^{1}\\
s_0 \ar[ur]^{\frac{1}{2}} \ar[dr]_{\frac{1}{2}} && s_3 \ar@(r,u)[]_{1}\\
& s_2 \ar[ur]_{1}
}
$$
extends the search.
\end{example}

Since the PTSs we will consider in the remainder of this paper all
extend a search $T$ of a PTS $\mathcal{S}$, we write $s_0$ instead
of ${s_0}_{\mathcal{S}}$ to avoid clutter.  
PTSs that extend a particular search give rise to the same set of 
execution paths if we restrict ourselves to those execution paths 
that only consist of transitions explored during the search.

\begin{proposition}
\label{prop:extension-execution}
If the PTS $\mathcal{S}'$ extends the search $T$ of the PTS
$\mathcal{S}$, then
\begin{itemize}
\item[(a)]
$T^* \cap \pref(\Exec_{\mathcal{S}}) = T^* \cap \pref(\Exec_{\mathcal{S}'})$
and 
\item[(b)]
$T^{\omega} \cap \Exec_{\mathcal{S}} = T^{\omega} \cap \Exec_{\mathcal{S}'}$.
\end{itemize}
\end{proposition}

PTSs that extend a particular search also assign the same measure
to basic cylinder sets of prefixes of execution paths only consisting 
of transitions explored during the search.

\begin{proposition}
\label{prop:extension-measure}
If the PTS $\mathcal{S}'$ extends the search $T$ of the PTS $\mathcal{S}$, 
then $\mu_{\mathcal{S}}(B_{\mathcal{S}}^e) = \mu_{\mathcal{S}'}(B_{\mathcal{S}'}^e)$
for all $e \in T^* \cap \pref(\Exec_{\mathcal{S}})$.
\end{proposition}

The function $\lab_{\mathcal{S}}$ assigns to each state the
set of atomic propositions that hold in the state.  This
function is extended to (prefixes of) execution paths as
follows.

\begin{definition}
The function $\trace_{\mathcal{S}} : \Exec_{\mathcal{S}} \to (2^{\AP_{\mathcal{S}}})^{\omega}$
is defined by
\begin{displaymath}
\trace_{\mathcal{S}}(t_1 t_2 \ldots) =
\lab_{\mathcal{S}}(\source_{\mathcal{S}}(t_1))
\lab_{\mathcal{S}}(\source_{\mathcal{S}}(t_2)) \ldots
\end{displaymath}
The function $\trace_{\mathcal{S}} : \pref(\Exec_{\mathcal{S}}) \to (2^{\AP_{\mathcal{S}}})^*$
is defined by
\begin{displaymath}
\trace_{\mathcal{S}}(t_1 \ldots t_n) =
\lab_{\mathcal{S}}(\source_{\mathcal{S}}(t_1)) \ldots
\lab_{\mathcal{S}}(\source_{\mathcal{S}}(t_n)) \lab_{\mathcal{S}}(\target_{\mathcal{S}}(t_n))
\end{displaymath}
\end{definition}

\begin{example}
Consider the PTS $\mathcal{S}$ of Example~\ref{example:first}.  
\begin{displaymath}
\begin{array}{rcl}
\trace_{\mathcal{S}}(t_{02} {t_{22}}^{\omega}) 
& = & \{ a \} \{ a, b \}^{\omega}\\
\trace_{\mathcal{S}}(t_{01} t_{13} {t_{33}}^{\omega}) 
& = & \{ a \} \{ a, b \} \{ a \}^{\omega}\\
\trace_{\mathcal{S}}(t_{01} t_{10} t_{02} {t_{22}}^{\omega})
& = & \{ a \} \{ a, b \} \{ a \} \{ a, b \}^{\omega}
\end{array}
\end{displaymath}
\end{example}

For the definition of linear-time property and the satisfaction
relation $\models$ we refer the reader to, for example, 
\cite[Section~3.2]{BK08}.  Based on these notions, we define
when an execution path of a PTS satisfies a linear-time
property.

\begin{definition}
The satisfaction relation $\models_{\mathcal{S}}$ is defined by
\begin{displaymath}
e \models_{\mathcal{S}} \phi \mbox{ if } \trace_{\mathcal{S}}(e)
\models \phi.
\end{displaymath}
\end{definition}

For PTSs that extend a particular search, those execution paths 
that only consist of transitions explored by the search satisfy
the same linear-time properties.

\begin{proposition}
\label{prop:extension-satisfaction}
Let $\phi$ be a linear-time property.  If the PTS $\mathcal{S}'$
extends the search $T$ of the PTS $\mathcal{S}$, then
$e \models_{\mathcal{S}} \phi$ iff $e \models_{\mathcal{S}'} \phi$
for all $e \in T^{\omega} \cap \Exec_{\mathcal{S}}$.
\end{proposition}
\begin{proof}
Since $\mathcal{S}'$ extends $T$ of $\mathcal{S}$,
$\trace_{\mathcal{S}}(e) = \trace_{\mathcal{S}'}(e)$ for all 
$e \in T^{\omega} \cap \Exec_{\mathcal{S}}$.
\qed
\end{proof}

Next, we introduce the notion of a progress measure.  Given a 
search of a PTS and a linear-time property, it captures the
amount of progress the search of the probabilistic model checker 
has made towards verifying the linear-time property.

\begin{definition}
Let the PTS $\mathcal{S}'$ extend the search $T$ of PTS $\mathcal{S}$
and let $\phi$ be a linear-time property.
The set $\mathcal{B}_{\mathcal{S}'}^{\phi}(T)$ is defined by
\begin{displaymath}
\mathcal{B}_{\mathcal{S}'}^{\phi}(T)
=
\bigcup \{\, B_{\mathcal{S'}}^e \mid e \in T^*
\wedge \forall e' \in B_{\mathcal{S}'}^e : e' \models_{\mathcal{S}'} \phi \,\}.
\end{displaymath}
\end{definition}

The set $\mathcal{B}_{\mathcal{S}'}^{\phi}(T)$ is the union of those
basic cylinder sets $B_{\mathcal{S'}}^e $ the execution paths of which
satisfy the linear-time property $\phi$.  Hence, $B_{\mathcal{S'}}^e $ 
does not contain any execution paths violating $\phi$. 
The set $\mathcal{B}_{\mathcal{S}'}^{\phi}(T)$ is measurable,
as shown in \cite[Proposition~1]{ZB11:icalp}.  Hence,
the measure $\mu_{\mathcal{S}'}$ assigns it a real number in the 
unit interval.  This number represents the ``size'' of the basic cylinder 
sets that do not contain any violations of~$\phi$.  This number captures 
the amount of progress of the search $T$ verifying $\phi$, 
\emph{provided that} 
the PTS under consideration is $\mathcal{S}'$.  However, we have no 
knowledge of the transitions other than the search.  Therefore, we consider 
all extensions $\mathcal{S}'$ of $T$ and consider the worst case in terms 
of progress.

\begin{definition}
The {\sl progress} of the search $T$ of the PTS $\mathcal{S}$ of 
the linear-time property $\phi$ is defined by
\begin{displaymath}
\prog_{\mathcal{S}}(T, \phi) 
=
\inf \left \{\, \mu_{\mathcal{S}'} \left (\mathcal{B}_{\mathcal{S}'}^{\phi}(T) \right ) \mid \mathcal{S}'
\mbox{ extends } T \mbox{ of } \mathcal{S} \right \}.
\end{displaymath}
\end{definition}

\begin{example}
Consider the PTS $\mathcal{S}$ of Example~\ref{example:first} and
the linear temporal logic formulae $\always a$, 
$\eventually a$,
$\eventually b$ and $\nxt b$.  In the table below, we present
the progress of these properties for a number of searches.
\begin{displaymath}
\begin{array}{|l|r|r|r|r|}
\hline
\mbox{\rm search}    & \always a & \eventually a & \eventually b & \nxt b\\
\hline
\emptyset            & 0         & 1             & 0             & 0\\
\hline
\{ t_{01} \}         & 0         & 1             & \frac{1}{2}   & \frac{1}{2}\\\hline
\{ t_{02} \}         & 0         & 1             & \frac{1}{2}   & \frac{1}{2}\\\hline
\{ t_{01}, t_{02} \} & 0         & 1             & 1             & 1 \\
\hline
\{ t_{01},t_{13}, t_{33} \}  & \frac{1}{4}    & 1           & \frac{1}{2}       & \frac{1}{2} \\
\hline
\{ t_{01},t_{10}, t_{13}, t_{33} \}  & \frac{1}{3}    & 1           & \frac{1}{2}       & \frac{1}{2} \\
\hline
\end{array}
\end{displaymath}
\end{example}

In \cite[Theorem~1]{ZB11:icalp}, Zhang and Van Breugel prove the
following key property of their progress measure.  They show that it is a
lower bound for the probability that the linear-time property holds.

\begin{theorem}
Let $T$ be a search of the PTS $\mathcal{S}$ and let $\phi$ be
a linear-time property.  Then
\begin{displaymath}
\prog_{\mathcal{S}}(T, \phi) 
\leq 
\mu_{\mathcal{S}}(\{\, e \in \Exec_{\mathcal{S}} \mid e \models_{\mathcal{S}}  \phi \,\}).
\end{displaymath}
\end{theorem}

The setting in this paper is slightly different from the one in
\cite{ZB11:icalp}.  In this paper we assume that PTSs do not have 
final states.  This assumption can be made without loss of any 
generality: simply add a self loop with probability one to
each final state.  

\section{Negation and Violations}

In this section, we consider the relationship between making progress
towards verifying a linear-time property and finding a violation of
its negation.  First, we formalize that a search has not found a 
violation of a linear-time property.

\begin{definition}
\label{def:violation}
The search $T$ of the PTS $\mathcal{S}$ {\sl has not found a violation} 
of the linear-time property $\phi$ if there exists a PTS $\mathcal{S}'$ 
which extends $T$ of $\mathcal{S}$ such that 
$e \models_{\mathcal{S}'} \phi$ for all $e \in \Exec_{\mathcal{S}'}$.
\end{definition}

This definition is slightly stronger than the one given in 
\cite[Definition~7]{ZB11:icalp}.  All results of \cite{ZB11:icalp} remain 
valid for this stronger version.  Next, we prove that if a search has 
made some progress towards verifying a linear-time property $\neg \phi$, then 
that search has also found a violation of $\phi$.

\begin{proposition}
Let $T$ be a search of the PTS $\mathcal{S}$ and let $\phi$ be a  
linear-time property.  If $\prog_{\mathcal{S}}(T,
\neg \phi) \gr 0$ then $T$ has found a violation of $\phi$.
\end{proposition}
\begin{proof}
By the definition of $\prog$,
$\mu_{\mathcal{S}'}(\mathcal{B}_{\mathcal{S}'}^{\neg \phi}(T)) \gr 0$
for each PTS $\mathcal{S}'$ which extends  $T$ of $\mathcal{S}$.  Hence, 
$\mathcal{B}_{\mathcal{S}'}^{\neg \phi}(T) \not= \emptyset$.
Therefore, there exists $e \in T^*$ such that
$B_{\mathcal{S}'}^e \not= \emptyset$ and 
$\forall e' \in B_{\mathcal{S}'}^e : e' \models_{\mathcal{S}'} \neg \phi$.
Hence, $e' \not\models \phi$ and $e' \in \Exec_{\mathcal{S}'}$.
Therefore, $T$ has found a violation of $\phi$.
\qed
\end{proof}

The reverse implication does not hold in general, as shown in the
following example.

\begin{example}
\label{example:pts}
Consider the PTS
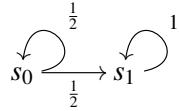
\begin{displaymath}
\UseComputerModernTips
\xymatrix{
s_0 \ar[r]_{\frac{1}{2}} \ar@(r,u)[]_{\frac{1}{2}} & s_1 \ar@(r,u)[]_{1}}
\end{displaymath}
Assume that the state $s_0$ satisfies the atomic proposition $a$
and the state $s_1$ does not.  Consider the linear-time property
$\always a$ and the search $\{ t_{00} \}$.  Note that 
${t_{00}}^{\omega} \not\models \neg \always a$ and, hence,
$\{ t_{00} \}$ has found a violation of $\neg \always a$.
Also note that $\prog_{\mathcal{S}}(\{ t_{00} \}, \always a)~=~0$.  
\end{example}

We conjecture that the reverse implication does hold for safety properties 
(see, for example, \cite[Definition~3.22]{BK08} for a formal definition 
of safety property).  However, so far we have only been able to 
prove it for invariants.

\begin{proposition}
If the search $T$ of the PTS $\mathcal{S}$ has found a violation of 
the invariant $\phi$ then 
\linebreak
$\prog_{\mathcal{S}}(T, \neg \phi) \gr 0$.
\end{proposition}
\begin{proof}
For every PTS $\mathcal{S}'$ that extends $T$,
$e \not\models_{\mathcal{S}'} \always a$ for some $e \in \Exec_{\mathcal{S}'}$.
Hence, $e = e_f t e_{\ell}$ for some 
$e_f \in T^* \cap \pref(\Exec_{\mathcal{S}'})$ and $t \in T$ such that
$a \not\in \lab_{\mathcal{S}'}(\source_{\mathcal{S}'}(t))$.
Therefore, for all $e' \in B_{\mathcal{S}'}^{e_f}$ we have that
$e' \models_{\mathcal{S}'} \neg \always a$ and
$B_{\mathcal{S}'}^{e_f} \not= \emptyset$.
Hence, $\mu_{\mathcal{S}'}(B_{\mathcal{S}'}^{e_f}) \gr 0$ and, 
therefore, $\prog_{\mathcal{S}}(T, \neg \always a) \gr 0$.
\qed
\end{proof}

\section{A Positive Fragment of LTL}

Next, we introduce a positive fragment of linear temporal logic (LTL).
This fragment lacks negation.  In Section~\ref{section:algorithm}
we will show how to compute the progress measure for this fragment.

\begin{definition}
The logic LTL$_+$ is defined by
\begin{displaymath}
\phi ::= \mbox{\rm true}  
\mid \mbox{\rm false}  
\mid a
\mid \phi \wedge \phi
\mid \phi \vee \phi 
\mid \nxt \phi 
\mid \phi_1 \until \phi_2
\mid \phi_1 \release \phi_2
\end{displaymath}
where $a \in \AP$.
\end{definition}

The grammar defining LTL$_+$ is the same as the grammar defining the logic 
PNF introduced in \cite[Definition~5.23]{BK08}, except that the
grammar of LTL$_+$ does not contain $\neg a$. For each LTL formula, there exists
an equivalent PNF formula (see, for example,
\cite[Section~5.1.5]{BK08}).  Such a result, of course, does not hold
for LTL$_+$.  

A property of LTL$_+$ that is key for our development is presented next.

\begin{proposition}
\label{prop:pnf}
For all LTL$_+$ formulae $\phi$ and $\sigma \in (2^{\AP})^*$,
$\sigma \emptyset^{\omega} \models \phi$
iff
$\forall \rho \in (2^{\AP})^{\omega} : \sigma \rho \models \phi$.
\end{proposition}
\begin{proof}
We prove two implications.  Let $\phi$ be a LTL$_+$ formula and
let $\sigma \in (2^{\AP})^*$.  Assume that 
\linebreak
$\forall \rho \in (2^{\AP})^{\omega} : \sigma \rho \models \phi$.
Since $\emptyset^{\omega} \in (2^{\AP})^{\omega}$, we can immediately
conclude that $\sigma \emptyset^{\omega} \models \phi$.

The other implication is proved by structural induction on $\phi$.
Let $\sigma \in (2^{\AP})^*$.  We distinguish the following cases.
\begin{itemize}
\item
In case $\phi = \mbox{\rm true}$, clearly 
$\forall \rho \in (2^{\AP})^{\omega} : \sigma \rho \models \phi$ and,
hence, the property is satisfied.
\item
In case $\phi = \mbox{\rm false}$, obviously
$\sigma \emptyset^{\omega} \models \phi$
is not satisfied and, therefore, the property holds.
\item
Let $\phi = a$.  If $\sigma \emptyset^{\omega} \models \phi$,
then $|\sigma| \gr 0$ and $a \in \sigma[0]$ and, hence,
$\forall \rho \in (2^{\AP})^{\omega} : \sigma \rho \models \phi$.
\item
Let $\phi = \phi_1 \wedge \phi_2$.  Assume that
$\sigma \emptyset^{\omega} \models \phi$.  Then
$\sigma \emptyset^{\omega} \models \phi_1$ and
$\sigma \emptyset^{\omega} \models \phi_2$.  By induction,
\linebreak
$\forall \rho \in (2^{\AP})^{\omega} : \sigma \rho \models \phi_1$
and $\forall \rho \in (2^{\AP})^{\omega} : \sigma \rho \models \phi_2$.
Hence, $\forall \rho \in (2^{\AP})^{\omega} : \sigma \rho \models \phi$.
\item
The case $\phi = \phi_1 \vee \phi_2$ is similar to the previous case.
\item
For $\nxt \phi$ we distinguish the following two cases.
Assume $|\sigma| = 0$.  Suppose 
$\sigma \emptyset^{\omega} \models \nxt \phi$.  Then
$\emptyset^{\omega}[1 \ldots] = \emptyset^{\omega} \models \phi$.   By induction,
$\forall \rho \in (2^{\AP})^{\omega} : \rho \models \phi$.
Hence, $\forall \rho \in (2^{\AP})^{\omega} : \rho \models \nxt \phi$.

Assume $|\sigma| \geq 1$.   Suppose 
$\sigma \emptyset^{\omega} \models \nxt \phi$.  Then
$(\sigma \emptyset^{\omega})[1 \ldots] = \sigma[1 \ldots]
\emptyset^{\omega} \models \phi$.  By induction,
\linebreak
$\forall \rho \in (2^{\AP})^{\omega} : \sigma[1 \ldots] \rho \models \phi$.
Since $\sigma[1 \ldots] \rho = (\sigma \rho)[1 \ldots]$, we
have that
$\forall \rho \in (2^{\AP})^{\omega} : \sigma \rho \models \nxt \phi$.


\item

Next, let $\phi = \phi_1 \until \phi_2$.  Assume that
$\sigma \emptyset^{\omega} \models \phi$.  Then there exists some 
$j \geq 0$ such that

\begin{itemize}
\item[(a)]
$(\sigma \emptyset^{\omega})[i \ldots] \models \phi_1$ for all  $0 \leq i \ls j$
and
\item[(b)]
$(\sigma \emptyset^{\omega})[j \ldots] \models \phi_2$ .
\end{itemize}

We distinguish two cases.  Suppose $j \ls |\sigma|$.  From (a) we can
conclude that for all $0 \leq i \ls j$, 
$(\sigma \emptyset^{\omega})[i \ldots] = \sigma[i \ldots] \emptyset^{\omega} \models \phi_1$.
By induction, 
$\forall \rho \in (2^{\AP})^{\omega} : \sigma[i \ldots] \rho \models \phi_1$.
Since $\sigma[i \ldots] \rho = (\sigma \rho)[i \ldots]$, we
have that
$\forall \rho \in (2^{\AP})^{\omega} : (\sigma \rho)[i \ldots] \models \phi_1$.
From (b) we can deduce that 
$(\sigma \emptyset^{\omega})[j \ldots] = \sigma[j \ldots] \emptyset^{\omega} \models \phi_2$.  By induction,
$\forall \rho \in (2^{\AP})^{\omega} : \sigma[j \ldots] \rho \models \phi_2$.
Since $\sigma[j \ldots] \rho = (\sigma \rho)[j \ldots]$, we
have that
$\forall \rho \in (2^{\AP})^{\omega} : (\sigma \rho)[j \ldots] \models \phi_2$.
Combining the above, we get
$\forall \rho \in (2^{\AP})^{\omega} : \sigma \rho \models \phi_1 \until \phi_2$.

Suppose  $j \geq |\sigma|$.   For $0 \leq i \ls |\sigma|$, the argument 
for (a) is the same as above.  For $|\sigma| \leq i \ls j$, (a) simply
says that $\emptyset^{ \omega} \models \phi_1$, which, by induction, implies that 
$\forall \rho \in (2^{\AP})^{\omega} : \rho \models \phi_1$.  Hence,
$\forall \rho \in (2^{\AP})^{\omega} : (\sigma \rho)[i \ldots] \models \phi_1$
for all $0 \leq i \ls j$.
In this case, (b) means $\emptyset^{\omega} \models \phi_2$, which, 
by induction, implies that
$\forall \rho \in (2^{\AP})^{\omega} :  \rho \models \phi_2$.
Hence,
$\forall \rho \in (2^{\AP})^{\omega} : (\sigma \rho)[j \ldots] \models \phi_2$.
Combining the above, we obtain that
$\forall \rho \in (2^{\AP})^{\omega} : \sigma \rho \models \phi_1 \until \phi_2$.


\item

Finally, we consider $\phi_1 \release \phi_2 $.  According to \cite[page~256]{BK08}, 
$\phi_1 \release \phi_2 \equiv \neg(\neg \phi_1 \until \neg \phi_2)$
and
\linebreak
$\neg (\phi_1 \until \phi_2) \equiv (\neg \phi_2) \wuntil (\neg \phi_1 \wedge \neg \phi_2)$.
According to \cite[page~252]{BK08}, 
$\phi_1 \wuntil \phi_2 \equiv (\phi_1 \until \phi_2) \vee \always \phi_1$.
Hence, we can derive that
$\phi_1 \release \phi_2 \equiv (\phi_2 \until (\phi_1 \wedge \phi_2))
\vee \always \phi_2$.  Therefore,  proving that the property is
satisfied by $\always \phi$, combined with the proofs for $\wedge$,
$\vee$ and $\until$ above, suffices as proof for $\phi_1 \release \phi_2$.

Thus, we consider $\always \phi$.  Suppose that
$\sigma \emptyset^{\omega} \models \always \phi$.   
Then $(\sigma \emptyset^{\omega})[j \ldots] \models \phi$
for all $j \geq 0$.  We distinguish two cases.
For all $0 \leq j \ls |\sigma|$, we have that
$(\sigma \emptyset^{\omega})[j \ldots] = \sigma[j \ldots]
\emptyset^{\omega} \models \phi$.
By induction, 
$\forall \rho \in (2^{\AP})^{\omega} : \sigma[j \ldots] \rho \models \phi$
and, hence,
$\forall \rho \in (2^{\AP})^{\omega} : (\sigma \rho)[j \ldots] \models \phi$.

For all $j \geq |\sigma|$, we have that
$(\sigma \emptyset^{\omega})[j \ldots] = \emptyset^{\omega} \models \phi$.
By induction, 
$\forall \rho \in (2^{\AP})^{\omega} : \rho \models \phi$ and,
therefore,
$\forall \rho \in (2^{\AP})^{\omega} : (\sigma \rho)[j \ldots] \models
\phi$.
Combining the above, we get
$\forall \rho \in (2^{\AP})^{\omega} : \sigma \rho \models \always \phi$.


\qed
\end{itemize}
\end{proof}

The above result does not hold for all LTL formulae, as shown in the
following example.

\begin{example}
Consider the LTL formula $\neg a$.  Note that this formula is not
equivalent to any LTL$_+$ formula.  Let $\sigma = \epsilon$.  Obviously,
$\emptyset^{\omega} \models \neg a$, but it is not the case that
$\forall  \rho \in (2^{\AP})^{\omega} : \rho \models \neg a$ (just
take a $\rho \in (2^{\AP})^{\omega}$ with $a \in \rho[0]$).
\end{example}

\section{An Algorithm to Compute Progress}
\label{section:algorithm}

To obtain an algorithm to compute the progress for the positive
fragment of LTL, we present an alternative characterization of
the progress measure.  This alternative characterization is
cast in terms of a PTS built from the search as follows.
We start from the transitions of the search and their source
and target states.  We add a sink state, which has a transition 
to itself with probability one and which does not satisfy any
atomic proposition.  For each state which has not been fully 
explored yet, that is, the sum of the probabilities of its 
outgoing transitions is less than one, we add a transition 
to the sink state with the remaining probability.  This PTS
can be viewed as the minimal extension of the search (we will
formalize this in Proposition~\ref{prop:characterization-cylinder}).
The PTS is defined as follows.

\begin{definition}
Let $T$ be a search of the PTS $\mathcal{S}$.
The set $S_{\mathcal{S}}^{T}$ is defined by
\begin{displaymath}
S_{\mathcal{S}}^{T}
=
\{\, \source_{\mathcal{S}}(t) \mid t \in T \,\} \cup
\{\, \target_{\mathcal{S}}(t) \mid t \in T \,\} \cup
\{ s_0 \}.
\end{displaymath}
For each $s \in S_{\mathcal{S}}^{T}$,
\begin{displaymath}
\out_{\mathcal{S}}(s) =
\sum \{\, \pro_{\mathcal{S}}(t) \mid t \in T \wedge \source_{\mathcal{S}}(t) = s \,\}.
\end{displaymath} 
The PTS $\mathcal{S}_T$ is defined by
\begin{itemize}
\item 
$S_{\mathcal{S}_T} = 
S^{T}_{\mathcal{S}} \cup \{ s_{\perp} \}$,
\item
$T_{\mathcal{S}_T} = 
T \cup \{\, t_s \mid s \in S^{T}_{\mathcal{S}} \wedge \out_{\mathcal{S}}(s) \ls 1 \,\} \cup \{ t_{\perp} \}$,
\item 
$\source_{\mathcal{S}_T}(t) =
\left \{
\begin{array}{ll}
\source_{\mathcal{S}}(t) & \mbox{if $t \in T$}\\
s & \mbox{if $t = t_s$}\\
s_{\perp} & \mbox{if $t = t_{\perp}$} 
\end{array}
\right .
$
\item 
$\target_{\mathcal{S}_T}(t) =
\left \{
\begin{array}{ll}
\target_{\mathcal{S}}(t) & \mbox{if $t \in T$}\\
s_{\perp} & \mbox{if $t = t_{\perp}$ or $t = t_s$}
\end{array}
\right .
$
\item 
$\pro_{\mathcal{S}_T}(t) =
\left \{
\begin{array}{ll}
\pro_{\mathcal{S}}(t) & \mbox{if $t \in T$}\\
1 - \out_{\mathcal{S}}(s) & \mbox{if $t = t_s$}\\
1 & \mbox{if $t = t_{\perp}$} 
\end{array}
\right .
$
\item
$\lab_{\mathcal{S}_T}(s) =
\left \{
\begin{array}{ll}
\emptyset & \mbox{if $s = s_{\perp}$}\\
\lab_{\mathcal{S}}(s) & \mbox{otherwise}
\end{array}
\right .
$
\end{itemize}
\end{definition}

The above definition is very similar to \cite[Definition~10]{ZB11:icalp}.
The main difference is that we do not have final states.

\begin{proposition}
Let $T$ be a search of the PTS $\mathcal{S}$.  Then the PTS
$\mathcal{S}_T$ extends $T$.
\end{proposition}
\begin{proof}
Follows immediately from the definition of $\mathcal{S}_T$.
\qed
\end{proof}

Next, we will show that the PTS $\mathcal{S}_T$ is the minimal
extension of the search $T$ of the PTS $\mathcal{S}$.  More
precisely, we will prove that for any other extension
$\mathcal{S}'$ of $T$ we have that
$\mu_{\mathcal{S}_T}(\mathcal{B}_{\mathcal{S}_T}^{\phi})
\leq
\mu_{\mathcal{S}'}(\mathcal{B}_{\mathcal{S}'}^{\phi})$.
To prove this result, we introduce two new notions and some
of their properties.

\begin{definition}
Let $T$ be a search of the PTS $\mathcal{S}$ and let $\phi$ be a 
linear-time property.
The set $\E_{\mathcal{S}}^{\phi}(T)$ is defined by
\begin{displaymath}
\E_{\mathcal{S}}^{\phi}(T)
=
\{\, e \in T^* \cap \pref(\Exec_{\mathcal{S}})
\mid \forall e' \in B_{\mathcal{S}}^e : e' \models_{\mathcal{S}} \phi \,\}.
\end{displaymath}
\end{definition}

The set $\E_{\mathcal{S}_T}^{\phi}(T)$ is minimal among
the $\E_{\mathcal{S}'}^{\phi}(T)$ where $\mathcal{S}'$
extends $T$.

\begin{proposition}
\label{prop:minimal}
Let the PTS $\mathcal{S}'$ extend the search $T$ of the PTS
$\mathcal{S}$.  For any LTL$_+$ formula $\phi$,
\linebreak
$\E_{\mathcal{S}_T}^{\phi}(T) \subseteq \E_{\mathcal{S}'}^{\phi}(T)$.
\end{proposition}

Next, we restrict our attention to those elements of
$\E_{\mathcal{S}}^{\phi}(T)$ which are minimal
with respect to the prefix order.

\begin{definition}
Let $T$ be a search of the PTS $\mathcal{S}$ and let $\phi$ be a 
linear-time property.
The set $\Emin_{\mathcal{S}}^{\phi}(T)$ is defined by
\begin{displaymath}
\Emin_{\mathcal{S}}^{\phi}(T)
=
\{\, e \in \E_{\mathcal{S}}^{\phi}(T)
\mid |e| \gr 0 \Rightarrow \exists e' \in B_{\mathcal{S}}^{e[|e|-1]} : e' \not\models_{\mathcal{S}} \phi \,\}.
\end{displaymath}
\end{definition}

Note that $e \in \Emin_{\mathcal{S}}^{\phi}(T)$ if and only if
it belongs to $\E_{\mathcal{S}}^{\phi}(T)$ and none of its
prefixes belong to $\E_{\mathcal{S}}^{\phi}(T)$.  

\begin{proposition}
\label{prop:cylinder}
Let the PTSs $\mathcal{S}'$ and $\mathcal{S}''$ extend the search $T$
of the PTS $\mathcal{S}$ and let $\phi$ be a linear-time property.  Then
\begin{displaymath}
\bigcup_{\bar{e} \in \Emin_{\mathcal{S}''}^{\phi}(T)} B_{\mathcal{S}'}^{\bar{e}}
=
\bigcup_{e \in \E_{\mathcal{S}''}^{\phi}(T)} B_{\mathcal{S}'}^e.
\end{displaymath}
\end{proposition}
\begin{proof}
Since $\Emin_{\mathcal{S}''}^{\phi}(T) \subseteq
\E_{\mathcal{S}''}^{\phi}(T)$, we can conclude that
the set on the left hand side is a subset of the set on the 
right hand side.  Next, we prove the other inclusion.  We show that
for each $e \in \E_{\mathcal{S}''}^{\phi}(T)$ there exists 
$\bar{e} \in \Emin_{\mathcal{S}''}^{\phi}(T)$ 
such that $B_{\mathcal{S}'}^{e} \subseteq B_{\mathcal{S}'}^{\bar{e}}$
by induction on the length of $e$.
In the base case, $|e| = 0$, then  $e \in \E_{\mathcal{S}''}^{\phi}(T)$
implies $e \in \Emin_{\mathcal{S}''}^{\phi}(T)$ and,
hence, we take $\bar{e}$ to be $e$.  Let $| e | \gr 0$.  We
distinguish two cases.  If 
$\exists e' \in B_{\mathcal{S}'}^{e[|e|-1]} : e' \not\models_{\mathcal{S}} \phi$
then we also take $\bar{e}$ to be $e$.  Otherwise, 
$e[|e|-1] \in \E_{\mathcal{S}''}^{\phi}(T)$.
Obviously, $B_{\mathcal{S}'}^{e} \subseteq B_{\mathcal{S}'}^{e[|e|-1]}$  
and, by induction, 
there exists a $\bar{e} \in \Emin_{\mathcal{S}''}^{\phi}(T)$ 
such that $B_{\mathcal{S}'}^{e[|e|-1]} \subseteq B_{\mathcal{S}'}^{\bar{e}}$.
\qed
\end{proof}

\begin{proposition}
\label{prop:additive}
Let the PTSs $\mathcal{S}'$ and $\mathcal{S}''$ extend the search $T$
of the PTS $\mathcal{S}$, and let $\phi$ be a linear-time property.  
If $\E_{\mathcal{S}'}^{\phi}(T) \subseteq \E_{\mathcal{S}''}^{\phi}(T)$
then
\begin{equation}
\label{eqn:additive}
\mu_{\mathcal{S}''}(\bigcup \{\, B_{\mathcal{S}''}^e \mid e \in \Emin_{\mathcal{S}'}^{\phi}(T) \,\})
=
\sum_{e \in \Emin_{\mathcal{S}'}^{\phi}(T)} \mu_{\mathcal{S}''}(B_{\mathcal{S}''}^e).
\end{equation}
\end{proposition}
\begin{proof}
We have that
\begin{eqnarray*}
\Emin_{\mathcal{S}'}^{\phi}(T)
& \subseteq & \E_{\mathcal{S}'}^{\phi}(T)
\comment{by definition}\\
& \subseteq & \E_{\mathcal{S}''}^{\phi}(T)
\comment{by assumption}\\
& \subseteq & \pref(\Exec_{\mathcal{S}''})
\comment{by definition}
\end{eqnarray*}
Hence, for all $e \in \Emin_{\mathcal{S}'}^{\phi}(T)$,
we have that $B_{\mathcal{S}''}^e \in \Sigma_{\mathcal{S}''}$.
Since the set $T$ is finite, the set $T^*$ is countable and, hence,
the set $\Emin_{\mathcal{S}'}^{\phi}(T)$ is countable 
as well.  Since a $\sigma$-algebra is closed under countable
unions, $\bigcup \{\, B_{\mathcal{S}''}^e \mid e \in
\Emin_{\mathcal{S}'}^{\phi}(T) \,\} \in \Sigma_{\mathcal{S}''}$.
Hence, the measure $\mu_{\mathcal{S}''}$ is defined on this set.

To conclude (\ref{eqn:additive}), it suffices to prove that for all 
$e_1$, $e_2 \in \Emin_{\mathcal{S}'}^{\phi}(T)$ 
such that $e_1 \not= e_2$, $e_1$ is not a prefix of $e_2$,
since this implies that $B_{\mathcal{S}'}^{e_1}$ and
$B_{\mathcal{S}'}^{e_2}$ are disjoint. 
Towards a contradiction, assume that $e_1$ is a prefix of $e_2$.  
Since $\forall e_1' \in B_{\mathcal{S}'}^{e_1} : e_1' \models_{\mathcal{S}'} \phi$ and
$e_1$ is a prefix of $e_2$ and $e_1 \not= e_2$, it 
cannot be the case that
$\exists e_2' \in B_{\mathcal{S}'}^{e_2[|e_2|-1]} : e_2' \not\models_{\mathcal{S}'} \phi$.
This contradicts the assumption that 
$e_2 \in \Emin_{\mathcal{S}'}^{\phi}(T)$.
\qed
\end{proof}

Now, we are ready to prove that the PTS $\mathcal{S}_T$ is the minimal
extension of the search $T$ of the PTS $\mathcal{S}$.

\begin{proposition}
\label{prop:characterization-cylinder}
Let the PTS $\mathcal{S}'$ extend the search $T$ of the PTS $\mathcal{S}$ 
and let $\phi$ be a LTL$_+$ formula.  Then
\begin{displaymath}
\mu_{\mathcal{S}_T}(\mathcal{B}_{\mathcal{S}_T}^{\phi}(T)) 
\leq 
\mu_{\mathcal{S}'}(\mathcal{B}_{\mathcal{S}'}^{\phi}(T)).
\end{displaymath}
\end{proposition}
\begin{proof}
\begin{eqnarray*}
\lefteqn{\mu_{\mathcal{S}_T}(\mathcal{B}_{\mathcal{S}_T}^{\phi}(T))}\\
& = & \mu_{\mathcal{S}_T}(\bigcup \{\, B_{\mathcal{S}_T}^e \mid e \in \E_{\mathcal{S}_T}^{\phi}(T) \,\})\\
& = & \mu_{\mathcal{S}_T}(\bigcup \{\, B_{\mathcal{S}_T}^e \mid e \in \Emin_{\mathcal{S}_T}^{\phi}(T) \,\})
\comment{Proposition~\ref{prop:cylinder}}\\
& = & \sum_{e \in \Emin_{\mathcal{S}_T}^{\phi}(T)} \mu_{\mathcal{S}_T}(B_{\mathcal{S}_T}^e)
\comment{Proposition~\ref{prop:additive}}\\
& = & \sum_{e \in \Emin_{\mathcal{S}_T}^{\phi}(T)} \mu_{\mathcal{S}'}(B_{\mathcal{S}'}^e)
\comment{Proposition~\ref{prop:extension-measure}}\\
& = & \mu_{\mathcal{S}'}(\bigcup \{\, B_{\mathcal{S}'}^e \mid e \in
\Emin_{\mathcal{S}_T}^{\phi}(T) \,\})
\comment{Proposition~\ref{prop:minimal} and \ref{prop:additive}}\\
& = & \mu_{\mathcal{S}'}(\bigcup \{\, B_{\mathcal{S}'}^e \mid e \in \E_{\mathcal{S}_T}^{\phi}(T) \,\})
\comment{Proposition~\ref{prop:cylinder}}\\
& \leq & \mu_{\mathcal{S}'}(\bigcup \{\, B_{\mathcal{S}'}^e \mid e \in \E_{\mathcal{S}'}^{\phi}(T) \,\})
\comment{Proposition~\ref{prop:minimal}}\\
& = & \mu_{\mathcal{S}'}(\mathcal{B}_{\mathcal{S}'}^{\phi}(T))
\end{eqnarray*}
\qed
\end{proof}

The above proposition gives us an alternative characterization
of the progress measure.

\begin{theorem}
\label{thm:characterization-progress}
Let $T$ be a search of the PTS $\mathcal{S}$ and let $\phi$ be a LTL$_+$ 
formula.  Then
\begin{displaymath}
\prog_{\mathcal{S}}(T, \phi) = \mu_{\mathcal{S}_T}(\mathcal{B}_{\mathcal{S}_T}^{\phi}(T)).
\end{displaymath}
\end{theorem}
\begin{proof}
This is a direct consequence of the definition of the progress measure
and Proposition~\ref{prop:characterization-cylinder}.
\qed
\end{proof}

Hence, in order to compute $\prog_{\mathcal{S}}(T, \phi)$, it
suffices to compute the measure of $\mathcal{B}_{\mathcal{S}_T}^{\phi}(T)$.
Next, we will show that the latter is equal to the measure of
the set of execution paths of $\mathcal{S}_T$ that satisfy $\phi$.
The proof consists of two parts.  First, we prove the following
inclusion.

\begin{proposition}
\label{prop:characterization-logical-inclusion}
Let $T$ be a search of the PTS $\mathcal{S}$ and let $\phi$ be a
linear-time property.  Then
\begin{displaymath}
\mathcal{B}_{\mathcal{S}_T}^{\phi}(T)
\subseteq
\{\, e \in \Exec_{\mathcal{S}_T} \mid e \models_{\mathcal{S}_T} \phi \,\}.
\end{displaymath}
\end{proposition}
\begin{proof}
Let $e \in \mathcal{B}_{\mathcal{S}_T}^{\phi}(T)$.
Then $e \in B_{\mathcal{S}_T}^{e'}$ for some $e' \in T^*$ such that
$\forall e'' \in B_{\mathcal{S}_T}^{e'} : e'' \models_{\mathcal{S}_T} \phi$.
Hence, $e \models_{\mathcal{S}_T} \phi$.
\qed
\end{proof}

The opposite inclusion does not hold in general, as shown in the
following example.

\begin{example}
Consider the PTS $\mathcal{S}$
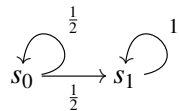
\begin{displaymath}
\UseComputerModernTips
\xymatrix{
s_0 \ar[r]_{\frac{1}{2}} \ar@(r,u)[]_{\frac{1}{2}} & s_1 \ar@(r,u)[]_{1}}
\end{displaymath}
Consider the search $\{ t_{00} \}$.  Then the PTS $\mathcal{S}_T$ can be depicted by
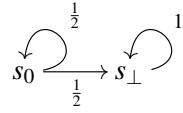
\begin{displaymath}
\UseComputerModernTips
\xymatrix{
s_0 \ar[r]_{\frac{1}{2}} \ar@(r,u)[]_{\frac{1}{2}} & s_{\perp} \ar@(r,u)[]_{1}}
\end{displaymath}
Assume that the state $s_0$ satisfies the atomic proposition $a$.  
Hence, ${t_{00}}^{\omega} \models_{\mathcal{S}_T} \always a$.
By construction, the state~$s_{\perp}$ does not satisfy $a$.
Therefore, ${t_{00}}^{\omega} \not\in \mathcal{B}_{\mathcal{S}_T}^{\always a}$.
\end{example}

However, we will show that the set
$\{\, e \in \Exec_{\mathcal{S}_T} \mid e \models_{\mathcal{S}_T} \phi \,\}
\setminus
\mathcal{B}_{\mathcal{S}_T}^{\phi}(T)$
has measure zero.  In the proof, we will use the following proposition.

\begin{proposition}
\label{prop:simplest-property}
Let $T$ be a search of the PTS $\mathcal{S}$ and let $\phi$ be a
linear-time property.  Assume that $T$ has not found a violation of $\phi$.
Then for all $e \in T^{\omega} \cap \Exec_{\mathcal{S}_T}$, 
$e \models_{\mathcal{S}_T} \phi$.
\end{proposition}
\begin{proof}
Let $e \in T^{\omega} \cap \Exec_{\mathcal{S}_T}$.  
Since $T$ has not found a violation of $\phi$, by definition
there exists a PTS $\mathcal{S}'$ that extends $T$ of $\mathcal{S}$
such that $e' \models_{\mathcal{S}'} \phi$ for all 
$e' \in \Exec_{\mathcal{S}'}$.  Then 
$e \in \Exec_{\mathcal{S}'} \cap T^{\omega}$ by 
Proposition~\ref{prop:extension-execution}(b),
because $\mathcal{S}'$ and $\mathcal{S}_T$ both extend $T$.
Hence, $e \models_{\mathcal{S}'} \phi$.  Therefore, from
Proposition~\ref{prop:extension-satisfaction} we can conclude 
that $e \models_{\mathcal{S}_T} \phi$.
\qed
\end{proof}

\begin{proposition}
\label{prop:characterization-logical-measure}
Let $T$ be a search of the PTS $\mathcal{S}$ and let $\phi$ be a
LTL$_+$ formula.  If $T$ has not found a violation of $\phi$ then
\begin{displaymath}
\mu_{\mathcal{S}_T}(\{\, e \in \Exec_{\mathcal{S}_T} \mid e \models_{\mathcal{S}_T} \phi \,\}
\setminus
\mathcal{B}_{\mathcal{S}_T}^{\phi}(T))
= 0.
\end{displaymath}
\end{proposition}
\begin{proof}
To avoid clutter, we denote the set
$\{\, e \in \Exec_{\mathcal{S}_T} \mid e \models_{\mathcal{S}_T} \phi \,\}
\setminus
\mathcal{B}_{\mathcal{S}_T}^{\phi}(T)$ by $Z$.

First, we show that $Z \subseteq T^{\omega}$.  Assume that $e \in Z$.
Towards a contradiction, suppose that $e \not\in T^{\omega}$.
From the construction of $\mathcal{S}_T$ we can deduce that
$e = e' t_s {t_{\perp}}^{\omega}$ for some $e' \in T^*$.  Let
$\trace_{\mathcal{S}_T}(e') = \sigma$.  Then 
$\trace_{\mathcal{S}_T}(e) = \sigma \emptyset^{\omega}$.
Since $e \in Z$, we have that
$e \models_{\mathcal{S}_T} \phi$ and, hence, 
$\sigma \emptyset^{\omega} \models \phi$.  By Proposition~\ref{prop:pnf},
$\forall \rho \in (2^{\AP})^{\omega} : \sigma \rho \models \phi$.
Hence, $\forall e'' \in B_{\mathcal{S}_T}^{e'} : e'' \models_{\mathcal{S}_T} \phi$.
Since $e \in B_{\mathcal{S}_T}^{e'}$, we have that 
$e \in \mathcal{B}_{\mathcal{S}_T}^{\phi}(T)$, which
contradicts our assumption that $e \in Z$.

Next, we show that each state in 
$\{\, \target_{\mathcal{S}_T}(e) \mid e \in \pref(Z) \,\}$ 
is transient.  Roughly speaking, a state $s$ is transient 
if the probability of reaching $s$ in one or more transitions when starting 
in $s$ is strictly less than one (see, for example, \cite[Section~7.3]{A70} 
for a formal definition).  It suffices to show that each state in 
$\{\, \target_{\mathcal{S}_T}(e) \mid e \in \pref(Z) \,\}$ can reach
the state $s_{\perp}$, since in that case the probability of reaching 
$s_{\perp}$ and, hence, not returning to the state itself, is greater than 
zero.

Since $T$ has not found a violation of $\phi$, we can conclude from
Proposition~\ref{prop:simplest-property} that
$e \models_{\mathcal{S}_T} \phi$ for all $e \in T^{\omega}$.  Hence,
from the construction of $\mathcal{S}_T$ we can deduce that if 
$e \not\models_{\mathcal{S}_T} \phi$ then $e \not\in T^{\omega}$ and,
hence, $e$ reaches~$s_{\perp}$.

Let $e \in \pref(Z)$.  Hence, there exists $e' \in
B_{\mathcal{S}_T}^e$ such that $e' \not\models_{\mathcal{S}_T} \phi$.
Therefore, $e'$ reaches $s_{\perp}$ and, hence, 
$\target_{\mathcal{S}_T}(e)$ can reach $s_{\perp}$.

Since $Z \subseteq T^{\omega}$, the set 
$\{\, \target_{\mathcal{S}_T}(e) \mid e \in \pref(Z) \,\}$ is finite.
According to \cite[page~223]{A70}, the probability of remaining in a finite set
of transient states is zero.  As a consequence, the probability of 
remaining in the set 
$\{\, \target_{\mathcal{S}'}(e) \mid e \in \pref(Z) \,\}$
is zero.  Hence, we can conclude that $\mu_{\mathcal{S}_T}(Z) = 0$.
\qed
\end{proof}

From the above, we can derive the following result.

\begin{theorem}
\label{thm:lowerbound}
Let $T$ be a search of the PTS $\mathcal{S}$ and let $\phi$ be a
LTL$_+$ formula.  If $T$ has not found a violation of $\phi$ then
\begin{displaymath}
\mu_{\mathcal{S}_T}(\mathcal{B}_{\mathcal{S}_T}^{\phi}(T))
=
\mu_{\mathcal{S}_T}(\{\, e \in \Exec_{\mathcal{S}_T} \mid e \models_{\mathcal{S}_T} \phi \,\}).
\end{displaymath}
\end{theorem}
\begin{proof}
\begin{eqnarray*}
\lefteqn{\mu_{\mathcal{S}_T}(\mathcal{B}_{\mathcal{S}_T}^{\phi}(T))}\\
& \leq &
\mu_{\mathcal{S}_T}(\{\, e \in \Exec_{\mathcal{S}_T} \mid e \models_{\mathcal{S}_T} \phi \,\})
\comment{Proposition~\ref{prop:characterization-logical-inclusion} and 
$\mu_{\mathcal{S}_T}$ is monotone}\\
& = & \mu_{\mathcal{S}_T}(\mathcal{B}_{\mathcal{S}_T}^{\phi}(T))
+ \mu_{\mathcal{S}_T}(\{\, e \in \Exec_{\mathcal{S}_T} \mid e \models_{\mathcal{S}_T} \phi \,\} \setminus \mathcal{B}_{\mathcal{S}_T}^{\phi}(T))
\comment{Proposition~\ref{prop:characterization-logical-inclusion} and 
$\mu_{\mathcal{S}_T}$ is additive}\\
& = & \mu_{\mathcal{S}_T}(\mathcal{B}_{\mathcal{S}_T}^{\phi}(T))
\comment{Proposition~\ref{prop:characterization-logical-measure}}
\end{eqnarray*}
\qed
\end{proof}

Combining Theorem~\ref{thm:characterization-progress} 
and \ref{thm:lowerbound}, we obtain the following characterization
of the progress measure.

\begin{corollary}
\label{cor:characterization-progress}
Let $T$ be a search of the PTS $\mathcal{S}$ and let $\phi$ be a
LTL$_+$ formula.  If $T$ has not found a violation of $\phi$ then
\begin{displaymath}
\prog_{\mathcal{S}}(T, \phi)
=
\mu_{\mathcal{S}_T}(\{\, e \in \Exec_{\mathcal{S}_T} \mid e \models_{\mathcal{S}_T} \phi \,\}).
\end{displaymath}
\end{corollary}
\begin{proof}
Immediate consequence of Theorem~\ref{thm:characterization-progress} 
and \ref{thm:lowerbound}.
\qed
\end{proof}

How to compute 
$\mu_{\mathcal{S}_T}(\{\, e \in \Exec_{\mathcal{S}_T} \mid e \models_{\mathcal{S}_T} \phi \,\})$
can be found, for example, in \cite[Section~3.1]{CY95:jacm}.  Computing
this measure is exponential in the size of $\phi$ and polynomial in 
the size of $T$.

\section{An Algorithm to Efficiently Compute a Lower Bound of Progress}

The algorithm developed in the previous section to compute 
$\prog_{\mathcal{S}}(T, \phi)$ is exponential in the size of~$\phi$.
In this section, we trade precision for efficiency.  We present
an algorithm that does not compute $\prog_{\mathcal{S}}(T, \phi)$, 
but only provides a lower bound in polynomial time.  This lower bound 
is tight for invariants.  However, we also show an example in which 
the lower bound does not provide us any information.  

Next, we show that subsets of $\Exec_{\mathcal{S}}$ can be
characterized as countable intersections of countable unions
of basic cylinder sets.  For $A \subseteq \Exec_{\mathcal{S}}$ 
and $n \in \Nset$, we use $A[n]$ to denote the set
$\{\, e[n] \mid e \in A \,\}$,
where $e[n]$ denotes the execution path $e$ truncated at length $n$.
We prove the characterization by showing two inclusions.  The
first inclusion holds for arbitrary subsets of $\Exec_{\mathcal{S}}$.

\begin{proposition}
\label{prop:subset-truncation}
For PTS $\mathcal{S}$, let $A \subseteq \Exec_{\mathcal{S}}$.  Then
\begin{displaymath}
A \subseteq \bigcap_{n \in \Nset} \bigcup_{e \in A[n]} B_{\mathcal{S}}^e.
\end{displaymath}
\end{proposition}
\begin{proof}
Let $e' \in A$.  It suffices to show that
\begin{equation}
\label{eq:subset-truncation}
e' \in \bigcup_{e \in A[n]} B_{\mathcal{S}}^e
\end{equation}
for all $n \in \Nset$.
Let $n \in \Nset$.  To prove (\ref{eq:subset-truncation}), it suffices
to show that $e' \in B_{\mathcal{S}}^e$ for some $e \in A[n]$.  Since $e' \in A$,
we have that $e'[n] \in A[n]$.  Because $e'[n]$ is a prefix of $e'$
and $e' \in \Exec_{\mathcal{S}}$, we have that 
$e' \in B_{\mathcal{S}}^{e'[n]}$, which concludes our proof.
\qed
\end{proof}

The reverse inclusion does not hold in general.  In some of the proofs
below we use some metric topology.  Those readers unfamiliar with
metric topology are referred to, for example, \cite{S75}.  To prove
the reverse inclusion, we use that the set is closed.

\begin{proposition}
\label{prop:superset-truncation}
For PTS $\mathcal{S}$, let $A \subseteq \Exec_{\mathcal{S}}$.  
If $A$ is closed then
\begin{displaymath}
\bigcap_{n \in \Nset} \bigcup_{e \in A[n]} B_{\mathcal{S}}^e \subseteq A.
\end{displaymath}
\end{proposition}
\begin{proof}
Let $e' \in \bigcap_{n \in \Nset} \bigcup_{e \in A[n]} B_{\mathcal{S}}^e$.
Then $e' \in \bigcup_{e \in A[n]} B_{\mathcal{S}}^e$ for all $n \in \Nset$.
Hence, for each $n \in \Nset$ there exists a $e_n \in A[n]$ such that
$e' \in B_{\mathcal{S}}^{e_n}$.  Thus, for each $n \in \Nset$ there exists 
a $e_n' \in A$ such that $e' \in B_{\mathcal{S}}^{e_n'[n]}$ and, hence,
$e_n'[n]$ is a prefix of $e'$.

We distinguish two cases.
Assume that for some $n \in \Nset$, $e_n'[n] = e_n'$.  Then
$e_n'$ is a prefix of $e'$.  Since also $e'$, 
$e_n' \in \Exec_{\mathcal{S}}$, we can conclude that
$e' = e_n'$.  Since $e_n' \in A$ we have that $e' \in A$.

Otherwise, $e_n'[n] \not= e_n'$ for all $n \in \Nset$.
Since also $e_n'[n]$ is a prefix of $e'$, we can conclude that
$e_n'[n] = e'[n]$. Let the distance function
$d : (\pref(\Exec_{\mathcal{S}}) \cup \Exec_{\mathcal{S}}) \times
(\pref(\Exec_{\mathcal{S}}) \cup \Exec_{\mathcal{S}}) \to [0, 1]$
be defined by $d(e_1, e_2) = \inf \{\, 2^{-n} \mid e_1[n] = e_2[n] \,\}$.
Then, $d(e_n', e') \leq 2^{-n}$, that is,
the sequence $(e_n')_n$ converges to $e'$.  Because all the
elements of the sequence $(e_n')_n$ are in $A$ and $A$ is closed,
we can conclude that the limit $e'$ is in $A$ as well (see, for
example, \cite[Proposition~3.7.15 and Lemma~7.2.2]{S75}).
\qed
\end{proof}

PTSs that extend a particular search assign the same measure
to closed sets of execution paths consisting only of explored
transitions.

\begin{proposition}
\label{prop:equal-for-closed-sets}
Let the PTS $\mathcal{S}'$ extend the search $T$ of the PTS $\mathcal{S}$
and let $A \subseteq T^{\omega} \cap \Exec_{\mathcal{S}}$.
If $A$ is closed then $\mu_{\mathcal{S}}(A) = \mu_{\mathcal{S}'}(A)$.
\end{proposition}
\begin{proof}
Obviously, for all $e \in T^*$ and $t \in T$, we have 
$B_{\mathcal{S}}^{e} \supseteq B_{\mathcal{S}}^{et}$.  As a consequence,
$\bigcup_{e \in A[n]} B_{\mathcal{S}}^e \supseteq \bigcup_{e \in A[n+1]} B_{\mathcal{S}}^e$
for all $n \in \Nset$.  Furthermore, 
$\mu_{\mathcal{S}}(\bigcup_{e \in A[0]} B_{\mathcal{S}}^e)
= \mu_{\mathcal{S}}(B_{\mathcal{S}}^{\epsilon}) 
= 1$ and, hence, $\mu_{\mathcal{S}}(\bigcup_{e \in A[0]} B_{\mathcal{S}}^e)$ 
is finite.  Since a measure is continuous (see, for example, 
\cite[Theorem~2.1]{B95}), we can conclude from the above that 
\begin{equation}
\label{eq:equal-for-closed-sets}
\mu_{\mathcal{S}} \left (\bigcap_{n \in \Nset} \bigcup_{e \in A[n]} B_{\mathcal{S}}^e\right )
=
\lim_{n \in \Nset} \mu_{\mathcal{S}} \left ( \bigcup_{e \in A[n]} B_{\mathcal{S}}^e \right ).
\end{equation}
Therefore,
\begin{eqnarray*}
\mu_{\mathcal{S}}(A)
& = & \mu_{\mathcal{S}} \left (\bigcap_{n \in \Nset} \bigcup_{e \in A[n]} B_{\mathcal{S}}^e \right )
\comment{Proposition~\ref{prop:subset-truncation} and \ref{prop:superset-truncation}}\\
& = & \lim_{n \in \Nset} \mu_{\mathcal{S}} \left ( \bigcup_{e \in A[n]} B_{\mathcal{S}}^e \right )
\comment{(\ref{eq:equal-for-closed-sets})}\\
& = & \lim_{n \in \Nset} \sum_{e \in A[n]} \mu_{\mathcal{S}}(B_{\mathcal{S}}^e)
\comment{a measure is countably additive}\\
& = & \lim_{n \in \Nset} \sum_{t_1 \ldots t_n \in A[n]} \prod_{1 \leq i \leq n} \pro_{\mathcal{S}}(t_i)\\
& = & \lim_{n \in \Nset} \sum_{t_1 \ldots t_n \in A[n]} \prod_{1 \leq i \leq n} \pro_{\mathcal{S}'}(t_i)
\comment{$\mathcal{S}'$ extends $T$ of $\mathcal{S}$}\\
& = & \mu_{\mathcal{S}'}(A)
\comment{by symmetric argument}.
\end{eqnarray*}
\qed
\end{proof}

Hence, the PTSs $\mathcal{S}$ and $\mathcal{S}_T$ assign the same
measure to the closed set of those execution paths consisting only 
of explored transitions.

\begin{corollary}
\label{cor:equal-measure}
Let $T$ be a search of the PTS $\mathcal{S}$.  Then
$\mu_{\mathcal{S}}(T^{\omega} \cap \Exec_{\mathcal{S}})
= \mu_{\mathcal{S}_T}(T^{\omega} \cap \Exec_{\mathcal{S}_T})$.
\end{corollary}
\begin{proof}
Since the sets $\Exec_{\mathcal{S}}$ and $T^{\omega}$ are closed, their
intersection is also closed (see, for example,
\cite[Proposition~3.7.5]{S75}) and, hence, the
result follows immediately from Proposition~\ref{prop:equal-for-closed-sets} 
and \ref{prop:extension-execution}(b).
\qed
\end{proof}

Now we can show that the measure of the set of execution paths 
consisting only of explored transitions is a lower bound for the
progress measure.

\begin{theorem}
Let $T$ be a search of the PTS $\mathcal{S}$ and let $\phi$ be a
LTL$_+$ formula.  If $T$ has not found a violation of $\phi$ then
\begin{displaymath}
\mu_{\mathcal{S}_T}(T^{\omega} \cap \Exec_{\mathcal{S}_T})
\leq
\prog_{\mathcal{S}}(T, \phi).
\end{displaymath}
\end{theorem}
\begin{proof}
\begin{eqnarray*}
\lefteqn{\mu_{\mathcal{S}_T}(T^{\omega} \cap \Exec_{\mathcal{S}_T})}\\
& \leq & \mu_{\mathcal{S}_T}(\{\, e \in \Exec_{\mathcal{S}_T} \mid e
\models_{\mathcal{S}_T} \phi \,\})
\comment{Proposition~\ref{prop:simplest-property}}\\
& = & \prog_{\mathcal{S}}(T, \phi)
\comment{Corollary~\ref{cor:characterization-progress}}
\end{eqnarray*}
\qed
\end{proof}

From the construction of $\mathcal{S}_T$ we can conclude that
$\mu_{\mathcal{S}_T}(T^{\omega} \cap \Exec_{\mathcal{S}_T})$ is the
same as 
\linebreak
$\mu_{\mathcal{S}_T}(\{\, e \in \Exec_{\mathcal{S}_T} \mid e
\mbox{ does not reach } s_{\perp} \,\})$, which is the same as
$1 - \mu_{\mathcal{S}_T}(\{\, e \in \Exec_{\mathcal{S}_T} \mid e
\mbox{ reaches } s_{\perp} \,\})$.  The latter can be computed
in polynomial time using, for example, Gaussian elimination 
(see, for example, \cite[Section~10.1.1]{BK08}).  This algorithm
has been implemented and incorporated into an
extension of the model checker JPF \cite{ZB10:qest}.  While
JPF is model checking sequential Java code which contains
probabilistic choices, our extension also keeps track of the 
underlying PTS.  The amount of memory needed to store this PTS 
is in general only a small fraction of the total amount of memory 
needed.  Once our extension of JPF runs almost out of memory, it
can usually free enough memory so that the progress can be computed 
from the stored PTS.

As was shown in \cite[Theorem~4]{ZB11:icalp}, the above bound
is tight for invariants.

\begin{proposition}
If the search $T$ of the PTS $\mathcal{S}$ has not found a violation 
of invariant $\phi$ then
\begin{displaymath}
\mu_{\mathcal{S}_T}(T^{\omega} \cap \Exec_{\mathcal{S}_T})
=
\prog_{\mathcal{S}}(T, \phi).
\end{displaymath}
\end{proposition}

In the example below, we present a search of a PTS for a LTL$_+$ formula 
of which the progress is one whereas the bound is zero.  In this case,
the bound does not provide us any information.

\begin{example}
Consider the PTS
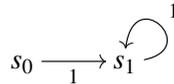
\begin{displaymath}
\UseComputerModernTips
\xymatrix{
s_0 \ar[r]_{1} & s_1 \ar@(r,u)[]_{1}}
\end{displaymath}
Assume that the state $s_1$ satisfies the atomic proposition $a$.
Consider the linear-time property $\nxt a$ and the search
$\{ t_{01} \}$.  In this case, we have that
$\prog_{\mathcal{S}}(\{ t_{01} \}, \nxt a) = 1$ but
$\mu_{\mathcal{S}_{\{ t_{01} \}}}(\{ t_{01} \}^{\omega} \cap \Exec_{\mathcal{S}_{\{ t_{01} \}}})
= \mu_{\mathcal{S}_{\{ t_{01} \}}}(\emptyset) = 0$.
\end{example}

\section{Conclusion}

Our work is based on the paper by Zhang and Van Breugel
\cite{ZB11:icalp}.  The work by Pavese, Braberman and
Uchitel \cite{PBU10:quovadis} is also related.  They aim to measure the
probability that a run of the system reaches a state
that has not been visited by the model checker.  Also
the work by Della Penna et al.\ \cite{DIMTZ06:ijsstt} seems related.
They show how, given a Markov chain and an integer $i$,
the probability of reaching a particular state $s$ within
$i$ transitions can be computed.

As we have seen, there seems to be a trade off between efficiency and
accuracy when it comes to computing progress.  Our algorithm to
compute $\prog_{\mathcal{S}}(T, \phi)$ is exponential in the size of 
the LTL$_+$ formula $\phi$ and polynomial in the size of the search $T$.  
We even conjecture (and leave it to future work to prove) that the 
problem of computing progress is PSPACE-hard.
However, in general the size of the LTL formula is small, whereas the 
size of the search is huge.  Hence, we expect our algorithm to be
useful.

Providing a lower bound for the progress measure can be done in
polynomial time.  As we have shown, this bound is tight for
invariants.  Invariants form an important class of properties.
Determining the class of LTL$_+$ formulae for which
the bound is tight is another topic for further research.

The approach to handle the positive fragment of LTL seems not applicable
to all of LTL.  We believe that a different approach is needed
and leave this for future research.

\smallskip

\noindent
\textbf{Acknowledgments}\ We thank the referees for their
constructive feedback.

\bibliographystyle{eptcs}
\bibliography{final}

\end{document}